\documentclass[letterpaper, 10 pt, conference]{ieeeconf}  

\IEEEoverridecommandlockouts                              
\usepackage{cdc-packages}
\usepackage{article-macros-technical}

\renewcommand{\em}{\it}
\def\fee{\upphi}
\newcommand{\RU}{R}
\newcommand{\RUinv}{\RU^{-1}}
\newcommand{\Ybar}{\bar{X}}

\newcommand{\etabar}{{\eta}}

\newcommand{\ud}{\mathrm{d}}
\def\Re{\mathbb{R}}
\def\R{\mathbb{R}}
\def\argmin{\mathop{\text{\rm arg\,min}}}

\def\Sec#1{Sec.~\ref{#1}}

\newcommand{\tr}{\mathrm{tr}}

\def\transpose{{\hbox{\rm\tiny T}}}

\newcommand{\cov}{\text{Cov}}
\newcommand{\backward}[1]{\overset{\shortleftarrow}{#1}}
\newcommand{\Ybarbar}{{Y}}
\newcommand{\SN}{S^{(N)}}
\def\Sec#1{Sec.~\ref{#1}}

\def\Prop#1{Prop.~\ref{#1}}

\newcommand{\val}{\tilde{\phi}}

\newcounter{rmnum}

\newcounter{anum}

\def\FRAC#1#2#3{\genfrac{}{}{}{#1}{#2}{#3}}
\def\half{{\mathchoice{\FRAC{1}{1}{2}}%
		{\FRAC{2}{1}{2}}%
		{\FRAC{3}{1}{2}}%
		{\FRAC{4}{1}{2}}}}

\newcommand{\inv}{^{-1}}
\newcommand{\tp}{^\transpose}
\newcommand{\normal}{\mathcal{N}}
\newcommand{\vf}{\mathcal{I}}
\newcommand{\correct}{\mathcal{C}}
\newcommand{\grad}{\nabla}
\newcommand{\pbart}{\bar{p}_t}
\newcommand{\pt}{p_t}
\newcommand{\bij}{\psi}
\newcommand{\vfc}{\mathcal{V}}
\newcommand{\vftr}{\omega}
\newcommand{\hjb}{\mathcal{H}}
\usepackage{booktabs}

\title{\LARGE \bf
Dual Ensemble Kalman Filter for Stochastic Optimal Control
}

\author{Anant A. Joshi, Amirhossein Taghvaei, Prashant G. Mehta, and Sean
  P. Meyn
\thanks{This work is supported in part by the AFOSR award
		FA9550-23-1-0060 (Joshi, Mehta), by the National Science Foundation (NSF) award EPCN-2318977  (Taghvaei) and by CCF 2306023 (Meyn).
	}
	\thanks{A.~A.~Joshi and P.~G.~Mehta are with the Coordinated
	Science Laboratory and the Department of Mechanical Science and
	Engineering at the University of Illinois at Urbana-Champaign
	(e-mail: anantaj2; mehtapg@illinois.edu).  A.~Taghvaei is with the
        Department of Aeronautics and Astronautics at the University of
        Washington Seattle (email: amirtag@uw.edu).
        S.~P.~Meyn is with
        the Department of Electrical and Computer Engineering at the
        University of Florida at Gainesville (email: meyn@ece.ufl.edu).}
}

\def\Ybar{{X}}

\def\transpose{{\intercal}}

\def\fee{\phi}

\begin{document}

\maketitle
\thispagestyle{empty}
\pagestyle{empty}

\begin{abstract}

In this paper, stochastic optimal control problems in
continuous time and space are considered.  
In recent years, such problems have
received renewed attention from the lens of reinforcement learning
(RL) which is also one of our motivation.   The main contribution is a
simulation-based algorithm -- dual ensemble Kalman
filter (EnKF) -- to numerically approximate the solution of
these problems.    
The paper extends our previous work where the dual EnKF was applied
in deterministic settings of the problem. 
The theoretical results and algorithms are illustrated with numerical experiments.

\end{abstract}

\section{Introduction}

Many types of reinforcement learning (RL) algorithms may be viewed as
``simulation-based'' where a model of a control system is simulated to
evaluate and iteratively improve a policy.  In
continuous-type continuous-space settings of this paper, the optimal
policy may be obtained from solving the Hamilton-Jacobi-Bellman (HJB)
equation for the value function (in linear Gaussian settings, the
equation reduces to a differential Riccati equation (DRE)).  A
simulation-based algorithm is useful for approximating the solution of
the HJB or the DRE for the cases where the state-space is too
large or the model parameters are not explicitly available (even
though a simulator for the same is).


In this paper, we consider optimal control problems where the the
control system is an It\^o stochastic differential equation (SDE) as
follows:
\begin{subequations}\label{eq:dyn}
\begin{align}
\ud X_t &= (a(X_{t}) + b(X_t) U_t)\ud t + \sigma(X_t) \ud W_t,\\
X_0&=x
\end{align}
\end{subequations}
\noindent where $X:=\{X_t:0\leq t\leq T\}$ is the $\Re^d$-valued state
process, $U:=\{U_t:0\leq t\leq T\}$ is the $\Re^m$-valued
control input, $W:=\{W_t:t\geq 0\leq t\leq T\}$ is a standard
Brownian motion (B.M.), and 
$a(\cdot),b(\cdot),\sigma(\cdot)$ are twice continuously
differentiable functions of their arguments.  The model is said to be
{\em deterministic} if $\sigma(x) = 0$ for all $x\in\Re^d$.  
The model is said to be linear Gaussian if $a(x) = Ax$, $b(x) = B$, and $\sigma(x) =
\sigma$.  

Our objective in this paper is to design a simulation-based algorithm
to approximate (or learn) the optimal control law.  Two types of
control objectives are considered for this purpose: (i) stochastic
optimal control (SOC); and (ii) risk sensitive control (RSC), both
over a finite-time horizon $[0,T]$.  The infinite-horizon case is
obtained by letting $T\to\infty$.  

The help explain the main idea of this paper, consider the SOC
problem.  Let $v_t(\cdot)$ denote the optimal value function for this
problem (for the linear Gaussian model with quadratic cost (LQG), the value function
$v_t(x) = x^\transpose P_t x + g_t$ is quadratic).  
Instead of
computing $v_t(x)$ by solving the HJB (or the DRE) equation,
our perspective is to view $\exp(- v_t(x))$ as an un-normalized
form of a probability density, denoted as $p_t(x)$. That is, 
\[
p_t(x) =
\frac{\exp(-v_t(x))}{\int_{\Re^d} \exp(-v_t(z)) \ud z}, \quad
x\in\Re^d, \; 0\leq t\leq T
\]
(For LQG, assuming $P_t$ is
invertible, $p_t = {\cal N}(0,
P_t^{-1})$ is a Gaussian density).  

\subsection{Contribution of this paper}

Our aim is to approximate the density $p_t(\cdot)$ as an ensemble
$\{Y_t^i\in\Re^d:   0\leq t\leq T, 1\leq i\leq N\}$ such that 
\begin{equation}\label{eq:yt_sample_from_p}
Y_t^i \stackrel{\text{i.i.d.}}{\sim} p_t(\cdot),\quad i=1,2,\hdots,N,
\quad 0\leq t\leq T
\end{equation}
(For LQG, $Y_t^i \stackrel{\text{i.i.d.}}{\sim} {\cal N}(0,
P_t^{-1})$ for $0 \le t \le T$).

The proposed simulation-based algorithm is a backward-in-time
controlled interacting particle system:
\begin{subequations}\label{eq:CIPS_intro}
\begin{align}
\ud Y^i_t &=  \underbrace{a(Y_t^i) \ud t + b\ud \backward{\eta}_t^i +
            \sigma(Y_t^i) \ud \backward{W}_t^i}_{i-\text{th copy of
    model}~\eqref{eq:dyn}} +  \underbrace{{\cal A}_t(Y_t^i; p_t^{(N)})
  \ud t}_{\text{mean field term}}\\
Y_T^i &\stackrel{\text{i.i.d.}}{\sim} \frac{\exp(-\mathcal{G}(\cdot))}{\int_{\R^d}\exp(-\mathcal{G}(z))\ud z},\quad i=1,2,\hdots,N,
\end{align}
\end{subequations}
where $\mathcal{G}$ is the terminal condition for the value function (for LQG it is a quadratic function of the state), and $p_t^{(N)}$ is the
empirical distribution of the ensemble.  The design problem is to design $\eta:=
\{(\eta_t^i, :0\leq
t\leq T\}$ and ${\cal A}:= \{{\cal A}_t:0\leq t\leq T\}$ such
that~\eqref{eq:yt_sample_from_p} holds in the mean-field limit
($N=\infty$). 

The first such algorithm appears in our prior
work~\cite{anant-2022} 
where explicit forms of ${\eta}$ and ${\cal
  A}$ are described for the deterministic setting of
the problem.  The
resulting algorithm is referred to as the dual ensemble Kalman filter (EnKF).  
The contribution of this paper is to extend the dual EnKF to
stochastic setting for the SOC and RSC problems. The specific types of cost structures
are introduced in~\Sec{sec:probform}.



\subsection{Relationship to literature}

The idea of transforming an optimal control problem into a sampling
problem is not new.  For the SOC problem, the idea has its roots in
the log transform which appears in~\cite{fleming-1982} and is related to the
minimum energy duality which is even
older~\cite{hijab-thesis,mortensen-1968}.  These types of
transformations are routinely re-discovered and have been applied for solving sampling/inference problems as optimal control
problems and vice-versa
(see~\cite{todorov-2007,kappen-2005,kappen-2005-jsm,rawlik2013stochastic,toussaint-2009,hoffman-2017}
and~\cite{levine-2018} for a review).  In the control
community, the moving horizon estimator is an example of the minimum
energy duality~\cite[Ch.~4]{rawlings-mayne}.

What is perhaps new in~\cite{anant-2022}  are two aspects of~\eqref{eq:CIPS_intro}:
\begin{enumerate}
\item Design of ${\eta}_t^i$ as a B.M.  This may be regarded as the
  exploration signal in RL.
\item Design of ${\cal A}_t(\cdot;\cdot)$.  This is the idea of
  designing interactions between simulations to approximate the
  solution of the HJB equation.
\end{enumerate}
In numerical performance and theoretical guarantees, these algorithms
can be order of magnitude better than the competing approaches
(detailed comparisons can be found in~\cite[Table 1]{anant-2022}).
There are some recent papers
from other groups~\cite{maoutsa-2022,reich-2024} which also use interacting particle systems to solve stochastic optimal control problems, a detailed comparison to which is presented in the main body of the paper.

It is noted that while the use of EnKF and controlled
interacting particle systems appears to be new for optimal control and
RL, EnKF is a standard simulation-based algorithm in the filtering
(data assimilation) applications~\cite{evensen2006}.

\subsection{Organisation of paper}

The outline of the remainder of this paper is as
follows. The mathematical formulations for the SOC and RSC problems
appear in~\Sec{sec:probform}.  The simulation-based algorithm to
approximate its solution appears in \Sec{sec:main-results}.  Each of
these sections include a discussion of historical and recent
literature on these topics.  The algorithms are numerically
illustrated in \Sec{sec:sim}. 


\section{Problem Formulation}
\label{sec:probform}

\textbf{Notation:} Given a symmetric positive semi-definite matrix $Q$, we let $|z|_{Q} \coloneqq z\tp Q z$. The normal distribution is represented as $\normal(\cdot,\cdot)$ where the first argument is the mean and second is the covariance. We let $\cov(\cdot)$ represent covariance.


\subsection{Deterministic Optimal Control (DOC)}
The simplest formulation is the deterministic case obtained when
$\sigma = 0$.  In this case, the SDE~\eqref{eq:dyn} reduces to an
ordinary differential equation and $X$ and $U$ are both deterministic
processes. 
The deterministic optimal control objective is as follows:
\[
J_T^{\text{\tiny DOC}}(U)  := \int_0^T \half \left(|c(X_t)|^2 + |U_t|^2_{\RU}\right) \ud t + \, \mathcal{G}(X_T) 
\]
where $c,\mathcal{G}$ are twice continuously differentiably real valued function, with $\mathcal{G}$ taking only non-negative value and $R$ is a symmetric strictly positive definite matrix.  

In a prior work~\cite{anant-2022}, a simulation-based algorithm for
approximation of the optimal control law is described.  The
algorithm is referred to as the dual ensemble Kalman filter or dual
EnKF for short. 
The goal of the present paper is to generalize and extend the dual EnKF to the stochastic optimal control problem for the system~\eqref{eq:dyn} when
$\sigma \neq 0$.  For this purpose, the following types of cost structures are considered.  

\subsection{Stochastic Optimal Control (SOC)}

The stochastic optimal control objective is as follows:
\[
J_T^{\text{\tiny SOC}}(U) := \mathbb{E} \left[ J_T^{\text{\tiny DOC}}(U)  \right]
\]
where randomness enters due to the Brownian motion in \eqref{eq:dyn}, and as noted earlier, we require the control to be adapted to the filtration generated by the Brownian motion. 

%

\subsection{Risk Sensitive Control (RSC)} 
\label{sec:rsc}
For risk sensitive control, the following objective is of interest \cite[Equation (1)]{nagai-2013}:
\[
J_T^{\text{\tiny RSC}}(U) = \theta^{-1}\log\mathbb{E} \left[ \exp{ ( \theta J_T^{\text{\tiny
        DOC}}(U) )} \right]
\]
where $\theta \in \Re \setminus \{0\} $ is referred to as the risk sensitive parameter.  The case $\theta > 0$ is known as risk averse and $\theta < 0$ as risk seeking. A rigorous treatment of this problem, along with its motivations, appear in \cite{nagai-2013}, \cite[Section 8]{fleming-chap6}. We note that the cost is always non-negative for every $\theta \ne 0$ because  $ \exp{ ( \theta J_T^{\text{\tiny
			DOC}}(U) )}>1$ if $\theta>0$ and $ \exp{ ( \theta J_T^{\text{\tiny
			DOC}}(U) )}<1$ when $\theta<0$ for all $U$.  

\subsection{Linear Quadratic (LQ) Control}

Suppose the system \eqref{eq:dyn} has linear time invariant dynamics (LTI), that is, $a(z) = Az$, $b(z)u = Bu$ and $\sigma(z) = \sigma$, and the cost has a quadratic structure, that is, $|c(z)|^2 = |Cz|^2$ and $\mathcal{G}(z) = |z|_{G}^2$ for matrices $A,B,\sigma,C,G$ where $G\tp = G \succ 0$. 

\begin{assumption}\label{assn:lq}
$(A,B)$ is controllable, $(A,C)$ is observable, and $B\RUinv B\tp - \theta\sigma\sigma\tp \succ 0$.
\end{assumption}
The last assumption is needed to ensure positive definiteness of the solution to the Riccati equation, which is introduced in subsequent sections  \cite[Equation (90)]{jacobson-1973}.
In this case SOC is called linear quadratic Gaussian (LQG) and RSC is called linear quadratic exponential Gaussian (LEQG). 

\subsection{Literature survey}

\textbf{SOC:} SOC is a standard and well studied problem in stochastic control theory \cite{astrom,bensoussan, doyle-1978}. Two approaches to solve it are using the stochastic maximum principle, and stochastic HJB equation \cite{bensoussan}. 


\textbf{RSC:} The RSC has a long history going back to early works like \cite{jacobson-1973} and \cite{whittle-1981}. A more recent introduction to LEQG can be found in \cite{nagai-2013}, \cite{fleming-chap6},\cite{fleming-2006}, and for some recent surveys on the topic, see \cite{whittle-2002}, \cite{basar-2021-survey}, \cite{borkar-2023-survey}. A stochastic maximum principle for LEQG was established in \cite{lim-2003}, and \cite{duncan-2013} solved the problem using fundamental ideas like completion of squares and Girsanov theorem. The connection of LEQG to games is also a well studied area, see for example,  \cite{jacobson-1973}, \cite{fleming-2006}, in which the connection to linear quadratic zero sum differential games is demonstrated, and \cite{james-1992} for how non-linear exponential cost problem relate to differential games.


\section{Proposed methodology}
\label{sec:main-results}

\begin{table}
\centering
\begin{tabular}{@{}cc@{}} 
& $\hjb(v)$ \\ \midrule
SOC & $\half |c|^2 + (\grad v)\tp a - \half (\grad v)\tp D \grad v + \half \tr(\Sigma \grad^2 v)$ \\ \midrule
RSC & $\half |c|^2 + (\grad v)\tp a - \frac{1}{2}(\grad v)\tp (D - \theta\Sigma) \grad v + \frac{1}{2} \tr(\Sigma \grad^2 v)$  \\ \bottomrule
\end{tabular}
\caption{The right-hand-side of the HJB equation~\eqref{eq:HJB} for the SOC and RSC problems. The notations $D \coloneqq bR\inv b\tp$ and $\Sigma \coloneqq \sigma\sigma\tp$.}
\label{tb:valuefn}
\end{table}

 The proposed methodology is presented in three subsections. In Sec.~\ref{sec:valuefn}, we relate the value function for the optimal control problems to a probability density function. In Sec.~\ref{sec:mean-field}, we relate the probability density function  to a suitably designed stochastic process that can be simulated as an interacting particle system. Finally, in Sec.~\ref{sec:LQ-result}, we present the specialization to the LQ setting, followed by a Gaussian approximation procedure that can be applied to nonlinear setup. 

\subsection{Transforming value function to probability density}
\label{sec:valuefn}
Consider the SOC and RSC problems. Value function $\{v_t(x): 0 \leq t\leq
T, x\in\Re^d\}$ is defined as the optimal cost-to-go over the horizon $[t,T]$ when the state $X_t=x$. 
According to the dynamic programming principle, the value function satisfies the Hamilton Jacobi Bellman (HJB) partial differential equation (PDE)
\begin{align}\label{eq:HJB}
	-\frac{\partial v}{\partial t} = \hjb(v),\quad v_T= \mathcal G
\end{align} 
where $\mathcal H$ is given in Table \ref{tb:valuefn}~\cite[Equation (3)]{nagai-2013},\cite[Equation (11.2.5)]{bensoussan}.
The optimal control is obtained as a function of $v$ as 
\begin{equation}
	U_t^* = -\RU\inv b(X_t)\tp v_t(X_t).
\end{equation}
%
In this paper, our strategy is to introduce a bijection $\bij : \Re
\to \Re$ so that 
\begin{align}\label{eq:pt}
p_t (x) : = \frac{\bij(v_t(x))}{\int \bij(v_t(x)) \ud x}, \quad 0\leq t\leq T,  \;x \in \Re^d
\end{align}
has the meaning of a probability density function.  
The bijection $\bij$ is selected according to
\begin{align}
\bij(z) \coloneqq 
\begin{cases}
\exp(-z); & \text{SOC} \\
\exp(-|\theta| z); & \text{RSC}
\end{cases}
\label{eq:bij}
\end{align}
Our choice of $\bij$ is inspired from the log transform which is routinely used in risk sensitive control \cite{fleming-chap6}. Since the value function is always positive, $\bij$ has been appropriately adjusted so that the quantity inside the exponential is negative. We also note that our method is only intended for situations when $\psi(v_t(\cdot)) \in L^1(\R)$ for all $0 \le t \le T$.



\begin{table}
\centering
\begin{tabular}{@{}cc@{}} \toprule
Cost & ${\mathcal{D}}(\Lambda)$ \\ \midrule
LQG & {$A^\transpose \Lambda + \Lambda A + C^\transpose C - \Lambda B  \RUinv B^\transpose \Lambda$} \\ \midrule
LEQG & $ A^\transpose \Lambda + \Lambda A + C^\transpose C  - \Lambda(B  \RUinv B^\transpose - \theta \sigma\sigma^\transpose ) \Lambda$ \\ \bottomrule
\end{tabular}
\caption{Expression for DRE in \eqref{eq:DRE}}
\label{tb:dre}
\end{table}

\subsection{Mean-field process}\label{sec:mean-field}
The goal is to design a (mean-field) stochastic process that has the same density function as  $\{p_t:0\leq t\leq T\}$. 

\begin{table}
\centering
\begin{tabular}{@{}cc@{}} \toprule
Cost & ${\mathcal{D}^{\dagger}}(\Lambda)$ \\ \midrule
LQG & {$A\Lambda + \Lambda A^\transpose - B\RU^{-1}B^\transpose + \Lambda C^\transpose C \Lambda$} \\ \midrule
LEQG & $ A\Lambda + \Lambda A^\transpose - \frac{1}{|\theta|}(B\RU^{-1}B^\transpose - \theta\sigma\sigma\tp) + |\theta|\Lambda C^\transpose C \Lambda$ \\ \midrule
\end{tabular}
\caption{Expression for the Dual DRE in \eqref{eq:d-DRE}}
\label{tb:d-dre}
\end{table}

\medskip

Define a stochastic process $\Ybarbar=\{\Ybarbar_t\in \Re^d:
0\leq t\leq T\}$ as a solution of the
following backward (in time) SDE:
\begin{subequations}\label{eq:Ybar}
\begin{align}
	\ud \Ybarbar_t & = a(\Ybarbar_t)\ud t  +  b(\Ybarbar_t)\ud  \backward{\etabar}_t +  \sigma(\Ybarbar_t)\ud  \backward{W}_t \nonumber \\ 
	& \qquad +( \vf_t({Y}_t; \pbart) + \correct_t({Y}_t; \pbart)) \ud t      \\ 
	\Ybarbar_T & \sim \frac{\bij(\mathcal{G}(\cdot))}{\int \bij(\mathcal{G}(z)) \ud z }
\end{align}
\end{subequations}
where ${\etabar}=\{{\etabar}_t \in \Re^m:0\leq t\leq T\}$ is a B.M. 
with a suitably chosen covariance matrix, $\vf_t(\cdot;\cdot), \correct_t(\cdot;\cdot)$ is a suitably chosen vector field, 
and $\bar{p}_t$ is the density of $\Ybarbar$. Here we have written the mean field term $\mathcal{A}$ introduced earlier as a sum of $\vf + \correct$, and the reason for which will be made clear later in this exposition.
The meaning of the backward
arrow on $\ud  \backward{\etabar}, \ud \backward{W}$ in~\eqref{eq:Ybar} is that the SDE
is simulated backward in time starting from the terminal condition
specified at time $t=T$.  The reader is referred
to~\cite[Sec. 4.2]{nualart1988stochastic} for the definition of the
backward It\^{o}-integral. 
The mean-field process is useful because of the following proposition.

\begin{proposition}
\label{prop:Y-exactness}
Consider the mean-field process~\eqref{eq:Ybar}.
Suppose $\cov(\eta)$ is Selected according to Table \ref{tb:soln},    $\vf$ and $\correct$ satisfy the PDEs
\begin{subequations}\label{eq:Poisson}
\begin{align}
-\nabla \cdot(\pbart(\cdot) \vf_t(\cdot;\pbart) ) &= \pbart(h_t(\cdot)-\hat{h}_t), \label{eq:PoissonI} \\
-\nabla \cdot(\pbart(\cdot) \correct_t(\cdot;\pbart) ) &= \vfc_t(\cdot), \label{eq:PoissonC}
\end{align} 
\end{subequations}
where $\vfc_t(\cdot)$, $h$
also appear in Table \ref{tb:soln}, and  $\hat{h}_t \coloneqq \int \pbart(z) h_t(z) \ud z$. 
Then, 
\begin{equation}
\bar{p}_t = p_t,\quad \forall t \in [0,T], 
\end{equation}
where $\pbart$ is the probability density function of $\Ybarbar_t$ and $p_t$ is defined in~\eqref{eq:pt} in terms of the value function. The optimal control is expressed as a function of $\pbart$ according to
\begin{align*}
U_t^* = 
\begin{cases}
R\inv b(X_t)\tp \grad\log\bar{p}_t(X_t); & \text{SOC} \\
(|\theta|R)\inv b(X_t)\tp \grad\log\bar{p}_t(X_t); & \text{RSC} \\
\end{cases}
\label{eq:opt-con}
\end{align*}
\end{proposition} 
\begin{proof}
See Appendix \ref{app:mf-proof}. 
\end{proof}
The significance of~\Prop{prop:Y-exactness} is that the optimal control policy $\phi_t(\cdot)$ can now be obtained in terms of the statistics of the random variable $\Ybarbar_t$. The PDEs written may not always be analytically tractable, in which case, one has to rely on numerical approximation techniques. We call $\vf$ as the interaction term, and $\correct$ as the correction term, since the latter accounts for non-constant system model. In other words, if $b,\sigma$ are constants, then 
$\correct$ becomes 0 for RSC with $\theta < 0$, and otherwise $\vfc$ simplifies to $\tr(\grad^2\pbart)$ for RSC with $\theta > 0$ and $\half\tr(\grad^2\pbart)$ for SOC.



\begin{table}
\centering
\begin{tabular}{@{}ccc@{}} \toprule
& $\vf_t^{\text{LQ}}(z; n_t, S_t)$ & $\correct_t^{\text{LQ}}(z; n_t, S_t)$ \\ \midrule
LQG& $\half S_t C\tp C(z + n_t) $ & $\half \Sigma S_t\inv (z - n_t)$  \\ \midrule
LEQG& \multirowcell{4}{$\frac{|\theta|}{2} S_t C\tp C(z + n_t)$} &  \multirowcell{2}{$\Sigma S_t\inv (z - n_t)$} \\ 
$\theta > 0$ & &  \\ \cmidrule{1-1} 
LEQG&  &  \multirowcell{2}{$0$}   \\ 
$\theta < 0$ & &  \\ \bottomrule
\end{tabular}
\caption{Vector fields for LQ case in \eqref{eq:EnKF-Yi}.}
\label{tb:soln-lq}
\end{table}

\begin{table*}
\centering
\begin{tabular}{@{}cccc@{}} \toprule
Cost & $\vfc(\cdot)$ & $h_t(\cdot)$ & $\cov(\eta)$ \\ \midrule
LQG & $\half \grad^2\cdot(\Sigma \pbart ) $ & $\half |c|^2 + \nabla \cdot a + \half \tr((D - \Sigma)\grad^2\log(\pbart))$  & {$\RUinv$}  \\ \midrule
LEQG & \multirowcell{2}{$ \grad\cdot(\pbart\grad\cdot (\Sigma - \frac{1}{\theta}D)) $} & \multirowcell{2}{$-\frac{\theta}{2} |c|^2 + \grad\cdot a + \half \grad^2 \cdot (\frac{1}{\theta}D - \Sigma) - \frac{1}{2\theta} \tr(D  \grad^2 \log \pbart)$} &  \multirowcell{2}{$(\sqrt{|\theta|}R)\inv$}  \\ 
$\theta < 0$ & & & \\ \midrule
LEQG & \multirowcell{2}{$ \grad^2\cdot(\pbart\Sigma) + \grad\cdot(\pbart \grad\cdot (\frac{1}{\theta}D - \Sigma) )$} & \multirowcell{2}{$\frac{\theta}{2} |c|^2 + \grad\cdot a - \grad^2\cdot(\frac{1}{\theta}D - \Sigma) + \half \tr((\frac{1}{\theta}D - 2\Sigma) \grad^2 \log \pbart) $} & \multirowcell{2}{$(\sqrt{|\theta|}R)\inv$}  \\ 
$\theta > 0$ & & & \\ \bottomrule
\end{tabular}
\caption{Details for the Poisson equation \eqref{eq:Poisson} in Proposition \ref{prop:Y-exactness}. }
\label{tb:soln}
\end{table*}



\subsection{LQ setting}\label{sec:LQ-result}
In this scenario, the value function is obtained as
\begin{align*}
v_t(x) = \half x\tp P_t x + g_t
\end{align*}
where $\{P_t:0\leq t\leq T\}$ is a matrix valued process which is the solution of the
backward (in time) {differential Riccati equation (DRE)}
\begin{equation}\label{eq:DRE}
-\frac{\ud}{\ud t} P_t = \mathcal{D}(P_t), \quad P_T = G
\end{equation}
where the expressions for the $\mathcal{D}(\cdot)$ are in Table \ref{tb:dre}, and 
\begin{align*}
-\dot{g}_t  = \tr(\sigma\sigma\tp P_t), \quad g_T = 0.
\end{align*} 
Under the Assumption \ref{assn:lq}, $P_t\succ 0$ for $0\leq t\le T$ whenever $G \succ
0$~\cite[Sec.~24]{brockett2015finite}, \cite[Equation (90)]{jacobson-1973}. Then density $p_t$ obtained from the value function then always takes the form $\normal(0,S_t)$, where
\begin{align*}
S_t := \begin{cases}
P_t\inv; \quad &\text{LQG} \\
(|\theta|P_t)^{-1}; \quad &\text{LEQG}
\end{cases}, 
\quad 0 \le t \le T.
\end{align*}
It is readily verified that $\{S_t:0\leq t\leq T\}$ also solves a
DRE 
\begin{equation}\label{eq:d-DRE}
-\frac{\ud}{\ud t} S_t = \mathcal{D}^{\dagger}(S_t),
\end{equation}
which represents the dual of~\eqref{eq:DRE},  and the expressions for $\mathcal{D}^{\dagger}$ appear in Table \ref{tb:d-dre}. The expressions for $\vf,\correct,\eta$ in this scenario are denoted by a superscript $^{\text{LQ}}$ and appear in Table \ref{tb:soln-lq}. The derivations of these expressions appear in Appendix \ref{app:mf-proof}.

The mean-field process is empirically approximated by simulating a
system of controlled interacting particles 
according to 
\begin{subequations}\label{eq:EnKF-Yi}
\begin{align}
&\ud {Y}^i_t = \underbrace{A {Y}^i_t \ud t + B\ud \backward{\eta}^i_t + \sigma \ud\backward{W}^i_t }_{\text{i-th copy of model}~\eqref{eq:dyn}} \nonumber\\ &+ 
\underbrace{\vf_t^{\text{LQ}}({Y}^i_t; n_t^{(N)}, S_t^{(N)}) + \correct_t^{\text{LQ}}({Y}^i_t; n_t^{(N)}, S_t^{(N)}) \ud t}_{\text{data assimilation process}} \\
&{Y}^i_T  \stackrel{\text{i.i.d}}{\sim} \mathcal
N(0,S_T),\quad 1\leq i\leq N, 
\end{align}
\end{subequations}
${\eta}^i$ and $W^i$ are an i.i.d copy of ${\etabar}$ and $W$ respectively,  $n^{(N)}_t := N^{-1}\sum_{i=1}^N
{Y}^i_t$, and 
\begin{align}\SN_t := \frac{1}{N-1}\sum_{i=1}^N
  ({Y}^i_t-n^{(N)}_t)({Y}^i_t-n^{(N)}_t)^\top \label{eq:SbarN}
\end{align}
The data assimilation process serves to couple the particles.  Without it, the particles are independent of each other. 
The finite-$N$ system~\eqref{eq:EnKF-Yi} is referred to as the {\em dual EnKF}.  
 
 \noindent \textbf{Optimal control:} Set $
\Ybar^i_t  := (\SN_t)^{-1} ({Y}^i_t  - n^{(N)}_t)$.  Define
\begin{align}
\tilde{K}^{(N)}_t \coloneqq
\begin{cases}
\dfrac{1}{N-1} \sum_{i=1}^N \Ybar^i_{t} (\Ybar^i_{t})^\transpose, &\, \, \text{LQG}  \\ \noalign{\vskip5pt}
\dfrac{1}{(N-1)|\theta|} \sum_{i=1}^N \Ybar^i_{t} (\Ybar^i_{t})^\transpose, &\, \,  \text{LEQG}  \\
\end{cases}
\label{eq:KtildeN}
\end{align}

We consider two cases: 
\begin{enumerate}
\item[(i)] If
the matrix $B$ is explicitly known then
\begin{equation}
K_t^{(N)}  = - \RU^{-1} B^\transpose\tilde{K}^{(N)}_t
\label{eq:KtN}
\end{equation}
from which the optimal control policy is approximated as
$
\phi_t^{(N)}(x) = K_t^{(N)}x.
$
\item[(ii)] If $B$ is unknown, define the Hamiltonian 
\begin{align*}
 H^{(N)}(x,\alpha,t) &:=
 \underbrace{\half
|Cx|^2 \quad + \half \alpha^\transpose R \alpha}_{\text{cost function}} \nonumber \\ & \quad +
 x^\transpose \tilde{K}^{(N)}_t \underbrace{(Ax + B\alpha)}_{\text{model}~\eqref{eq:dyn}} \label{eq:hamN}
\end{align*}
from which the
  optimal control policy is approximated as
\[
\phi_t^{(N)}(x) = \argmin_{a\in\Re^m} H^{(N)}(x,a,t)
\]
We obtain $(Ax + B\alpha)$ by
averaging the model \eqref{eq:dyn} 
$N_s$ many times. 
The error in estimation would scale as $1/N_s$ from the strong law of large numbers. There are several zeroth-order approaches to solve the minimization problem, e.g., by
constructing 2-point estimators for the
gradient.  Since the
objective function is quadratic and the matrix $R$ is known, $m$ queries
of $H^{(N)}(x,\cdot,t)$ are sufficient to compute $\fee_t^{
  (N)}(x)$. 
\end{enumerate}
The interacting particle system \eqref{eq:EnKF-Yi} is simulated using Euler-Maruyama discretization scheme, where the direction of time is reversed. The discretization scheme is similar to the dual EnKF algorithm that appears in \cite{anant-2022}.

The correction and interaction terms are simplified under certain assumption about the model, as described in the following result. 
\begin{proposition}
\label{prop:Bsigma}
Consider the mean-field process in the LQ setting. 
\begin{enumerate}
\item For LQG, if $BR\inv B\tp = \sigma\sigma\tp + B\tilde{R}B\tp$ for some $\tilde{R} \succeq 0$, then set $\cov({\eta}) = \tilde{R}$ and $\correct_t \equiv 0$. In particular, if $\sigma = BR^{-\half}$, then $\eta \equiv 0$.
\item For LEQG with $\theta > 0$, if $BR\inv B\tp = 2\theta\sigma\sigma\tp + \theta B\tilde{R}B\tp$ for some $\tilde{R} \succeq 0$, then set $\cov({\eta}) = \tilde{R}$ and $\correct_t \equiv 0$. In particular, if $\sigma = B(2\theta R)^{-\half}$, then $\eta \equiv 0$.
\end{enumerate}
\end{proposition} 
\begin{proof}
See Appendix \ref{app:mf-proof}.
\end{proof}

\begin{remark}
\label{rmk:model-free}
We make some observations on making our algorithm model-free. In the earlier work \cite{anant-2022}, emphasis was on designing algorithms which can be implemented without having access to the model parameters in \eqref{eq:dyn}, but with only access to model evaluations. 
The LEQG for $\theta < 0$ is model free, since the vector field $\vf^{\text{LQ}}$ does not involve any of the model parameters $A,B,\sigma$. 
Similarly, the situations considered in Proposition \ref{prop:Bsigma} can be implemented in a model free manner.
\end{remark}

\subsection{Gaussian Approximation}
For a numerical approximation of  the solution of the Poisson equations, we notice that the terms simplify in the following case:
\begin{enumerate}
\item $a(x)$ is conservative, i.e. $\nabla\cdot a(x) = 0$.  
\item If $b(x)=B$ then $\nabla \cdot D = 0$.
\item If $\sigma(x) = \sigma$ then $\nabla \cdot \Sigma = 0$.
\end{enumerate}    
Making these simplifications, 
and considering a Gaussian approximation for the density
$p_t$ we get 
\begin{align*}
h_t(x) = \frac{1}{2}|c(x)|^2 +
\text{(constant)}
\end{align*}  
which makes the solution of \eqref{eq:Poisson} simpler. 
This is useful to obtain a dual EnKF
algorithm:
\begin{align*}
&\ud {Y}^i_t = \underbrace{a ({ Y}^i_t)\ud t +  b ({ Y}^i_t) \ud \backward{\eta}^i_t + \sigma \ud\backward{W}^i_t }_{\text{i-th copy of model}~\eqref{eq:dyn}} \nonumber\\ 
& + 
\underbrace{\vf_t^{\text{GA}}({Y}^i_t; n_t^{(N)}, S_t^{(N)}) + \correct_t^{\text{LQ}}({Y}^i_t; n_t^{(N)}, S_t^{(N)}) \ud t}_{\text{data assimilation process}} \\
&{Y}^i_T  \stackrel{\text{i.i.d}}{\sim} \mathcal
N(0,S_T),\quad 1\leq i\leq N, 
\end{align*}
where the vector field $\correct^{\text{LQ}}$ is the same as the linear quadratic case (Table \ref{tb:soln-lq}), $\eta^i:=\{\eta_t^i \in \Re^m : i:0\leq t\leq T\}$ is an independent copy of $\etabar$,  and 
\begin{align*}
\vf_t^{\text{GA}}({Y}^i_t; n_t^{(N)}, S_t^{(N)}) := 
\begin{cases}
\frac{1}{2(N-1)}\tilde{\vf}; & \text{SOC} \\ \noalign{\vskip5pt}
\frac{1}{2|\theta|(N-1)}\tilde{\vf}; & \text{RSC}
\end{cases}
\end{align*}
where
\begin{align*}
\tilde{\vf} :=
& \sum_{j=1}^N  ({ Y}^j_t-n^{(N)}_t )(c({ Y}^j_t) - \hat{c}_t^{(N)}  )^\top (c(Y_t^i) + \hat{c}^{(N)})
\end{align*}
and $\hat{c}^{(N)}_{t} := N^{-1} \sum_{i=1}^N c(\Ybarbar^i_{t})$.
One may interpret the above as the dual counterpart of the FPF
algorithm with a constant gain approximation~\cite[Example 2]{yang-2016}. 

The optimal control may be approximated as earlier via the Hamiltonian (with $\tilde{K}_t^{(N)}$ as in \eqref{eq:KtildeN}),
\begin{align*}
H^{(N)}(x,\alpha,t) &:=  \half
|c(x)|^2 + \half \alpha\tp R \alpha \\ 
& \qquad +  x\tp\tilde{K}_t^{(N)}\underbrace{(a(x) + b(x)\alpha)}_{i-\text{th copy of model}~\eqref{eq:dyn}}  
\end{align*}
where as before $X_t^i \coloneqq (S^{(N)}_t)^{-1}(Y_t^i - n^{(N)}_t)$.  


\subsection{Comparison to literature}
\label{sec:comparison}

In recent years, there has been work on using particle based methods to approximate the solution of Fokker Planck equation \cite{maoutsa-2020},  and the solution of stochastic non-linear affine and quadratic in control problems \cite{maoutsa-2022}, \cite{reich-2024}. In \cite{maoutsa-2020}, authors express the Fokker Planck equation as a Liouville equation by incorporating the score function (defined as the gradient of log of density) in the dynamics of the original system. Then they adopt a variational representation for the score function and propose a particle based approach to estimate it. 

Closely related to our approach are \cite{maoutsa-2022}, \cite{reich-2024} which are also based on the same fundamental idea of turning value functions to probability density functions using the exponential transform. 
However, in \cite{maoutsa-2022} and \cite{reich-2024}, the density obtained from the value function $p$ is expressed as ratio of two densities $q$ and $\rho$, which is very much like smoothing \cite{jinkim-2020-md}. To be precise, they write $q_t(x) = \rho_t(x) p_t(x)$, where $\rho$ and $q$ propagates forward in time. The PDEs governing the two densities are as follows
\begin{align*}
\frac{\partial \rho_t(x)}{\partial t} &= -\grad\cdot(a\rho_t) + \tr(\Sigma\grad^2\rho_t) - \rho_tc \\
\frac{\partial q_t(x)}{\partial t} &= -\grad\cdot((a + bU_t^*)\rho_t ) + \tr(\Sigma\grad^2q_t)
\end{align*}
and then they replace $U_t^* = \RU\inv b\tp\grad\log p_t = \RU\inv b\tp(\grad\log q_t - \grad\log\rho_t)$.
The two densities $\rho$ and $q$ are each simulated as a coupled interacting particle system that involves approximation of the so-called score function, i.e. gradient of the logarithm of density. Although the PDEs for the densities they obtain are very similar to the PDEs that we get (see Table \ref{tb:pt}), a significant difference is that they need two separate interacting particle systems and their simulation for $q$ utilises the result of the simulation for $\rho$. 
In contrast, our proposed approach only involves one interacting particle system and is also applicable to risk sensitive case. Although we divide our solution into two parts $\vf$ and $\correct$, we do it for convenience of notation and comprehension, and they may very well be combined into a single vector field by summing them. 
Moreover, we lay emphasis on the model-free setting under certain scenarios as stated in Remark \ref{rmk:model-free} while their approach needs the model.

{The idea of the log transform employed in this paper builds on our previous work \cite{anant-2022}, in which we make detailed comparisons to the body of literature in nonlinear filtering theory. In a companion paper \cite{arxiv-leqg} we show error analysis with respect to the number of particles $N$ for the linear quadratic case wherein we quantify the error in estimating $S_t^{(N)}$, which changes as $1/N$. We also compare our approach to policy gradient type reinforcement learning approaches which solve the same problem.
We would like to emphasize here that we employ a simulator based approach, as opposed to a model identification approach based on estimating the system dynamics and solving the HJB equation. For the general nonlinear case, our work involves solving the Poisson equation \eqref{eq:Poisson}, see e.g. \cite{amir-2020-siam}.}  

\section{Numerics}
\label{sec:sim}

In this section we present results  from numerical experiments. We evaluate our algorithm on two examples: the inverted pendulum on cart system (Figure \ref{fig:ivp}) and a spring mass damper system with one mass (Figure \ref{fig:smd}). The models for both systems appear in \cite{anant-2022}.
We let LEQGP represent the case with $\theta > 0$ and LEQGN the case with $\theta < 0$.

For the inverted pendulum on cart, the objective is to stabilise the system at $\theta = \pi$ and $x = 0$ where $\theta$ is the angle the pendulum makes with the downward vertical direction, and $x$ is the displacement of the cart from a desired equilibrium. We use the Gaussian approximation method to implement the dual EnKF. algorithm We observe that all three controllers are able to stabilise the system reasonably well at the desired equilibria.

The spring mass damper, is a linear system on which we apply a quadratic cost, and we show convergence of the EnKF output to the solution of the respective ARE.  

\begin{figure}
\centering
\includegraphics[scale=0.25]{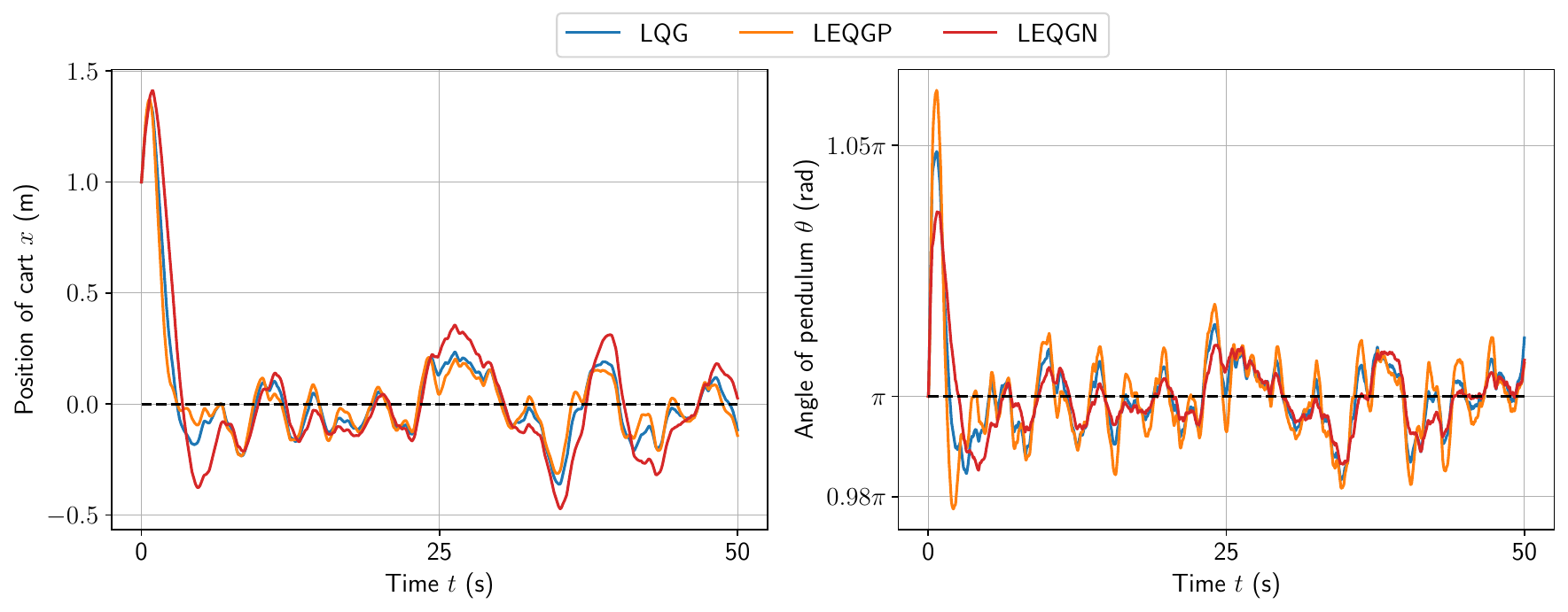}
\caption{Performance of all three algorithms on inverted pendulum on cart. The task is to stabilise the system state at $x=0$ and $\theta = \pi$. }
\label{fig:ivp}
\end{figure}

\begin{figure}
\centering
\includegraphics[scale=0.25]{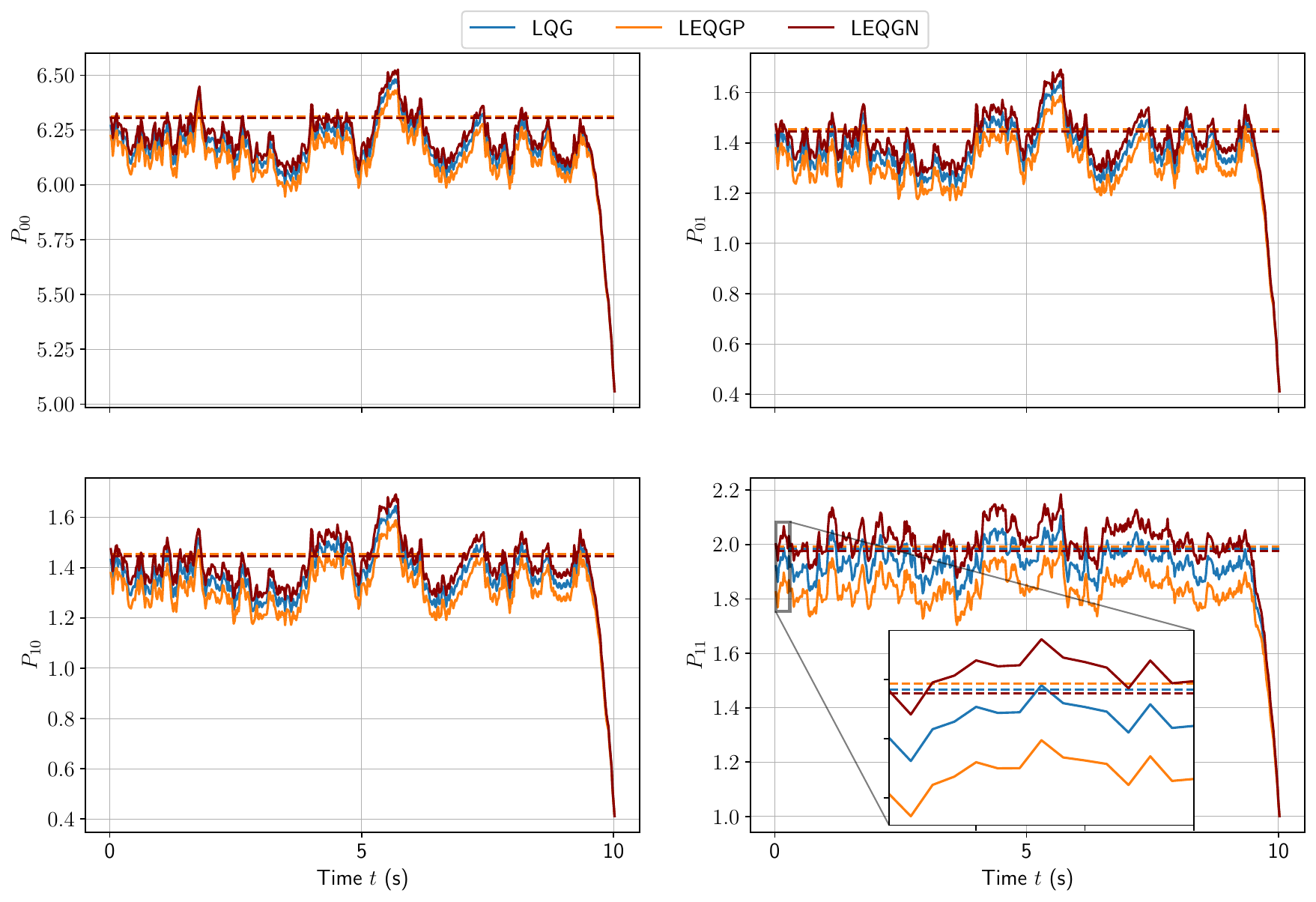}
\caption{Performance of all three algorithms on spring mass damper. The dashed lines represent the solutions of the respective AREs, and the solid lines are the solutions obtained by running the EnKF algorithm.}
\label{fig:smd}
\end{figure}


\bibliography{refs,litsur,literature}
\bibliographystyle{IEEEtran}

\appendices

\section{Proof of main results}
\label{app:mf-proof}

\begin{table*}[!b]
\centering
\begin{tabular}{@{}cp{10cm}@{}} \toprule
Cost & \hspace{4cm} $\frac{\partial \pt}{\partial t} $ \\ \midrule
\multirowcell{2}{SOC} & 
$=\pt(\dot{\beta}_t + \half |c|^2) -(\grad p)\tp a + \frac{1}{2\pt} \grad \pt \tp (\Sigma - D ) \grad \pt - \frac{1}{2}\tr(\Sigma \grad^2 \pt)$ \\
&\vspace{-5pt}$=\pt(h_t - \hat{h}_t) -\nabla \cdot (\pt a) - \frac{1}{2}\nabla^2 \cdot (D\pt) +\grad \cdot (\pt\grad\cdot D)$ \\
 \midrule
\multirowcell{2}{RSC \\ $\theta > 0$} & 
$= \pt(\dot{\beta}_t + \frac{\theta}{2} |c|^2) - (\grad \pt)\tp a + \frac{1}{2\pt}(\grad \pt)\tp (2\Sigma - \frac{1}{\theta}D ) \grad \pt - \frac{1}{2} \tr(\Sigma \grad^2 \pt)$ \\
& \vspace{-5pt}$= \pt(h_t -\hat{h}) - \nabla \cdot (\pt a) - \frac{1}{2} \grad^2\cdot(\pt(\frac{1}{\theta}D - \Sigma)) + \grad\cdot(\pt \grad\cdot (\frac{1}{\theta}D - \Sigma) )$ 
 \\ \midrule
\multirowcell{2}{RSC \\ $\theta < 0$} & 
$= \pt(\dot{\beta}_t -\frac{\theta}{2} |c|^2) - (\grad \pt)\tp a + \frac{1}{2\theta \pt}(\grad \pt)\tp D \grad \pt - \frac{1}{2} \tr(\Sigma \grad^2 \pt) $ \\
&\vspace{-5pt} $= \pt(h - \hat{h}) - \nabla \cdot (\pt a) - \frac{1}{2} \grad^2\cdot(\pt(\Sigma - \frac{1}{\theta}D)) + \grad\cdot(\pt\grad\cdot (\Sigma - \frac{1}{\theta}D)) $ \\ \bottomrule
\end{tabular}
\caption{PDE for density obtained from \eqref{eq:pt}}
\label{tb:pt}
\end{table*}

The PDE for $p_t(\cdot)$, defined in \eqref{eq:pt}, are expressed in Table \ref{tb:pt}.
Let $R_{\eta}$ denote the covariance of $\eta$, and define $D_{\eta}(z) \coloneqq b(z)R_{\eta}\inv b(z)\tp$. Then the Fokker Planck equation for the mean field system \eqref{eq:Ybar} reads
\begin{align} 
\frac{\partial \pbart}{\partial t} &= - \nabla \cdot (\pbart (a + \vf_t(\cdot;\pbart) + \correct_t(\cdot;\pbart))) \nonumber \\ & \qquad - \half \nabla^2 \cdot (\pt(D_{\eta} + \Sigma)).
\label{eq:FPE}
\end{align}

\textbf{Proof of Proposition \ref{prop:Y-exactness}:} 
Substituting the appropriate vector fields from Table \ref{tb:soln} into \eqref{eq:FPE}, we see that $\pt(\cdot)$ and $\pbart$ satisfy the same PDEs with the same terminal conditions. We used the following expressions to do the calculations
\begin{align*}
a\tp \grad \pbart &= \grad\cdot(\pbart a) - \pbart\grad\cdot a \\
\tr(\Sigma\grad^2 \pbart) &= \grad^2\cdot(\pbart\Sigma) - 2\grad\cdot(\pbart\grad\cdot \Sigma) + \pbart \grad^2 \cdot \Sigma \\
\grad^2\log \pbart &= \frac{1}{\pbart}\grad^2\pbart - \frac{1}{\pbart^2}\grad \pbart\grad \pbart\tp, \\
\beta_t &\coloneqq \int \pbart(z) \ud z, \quad
\hat{h} \coloneqq \int h(z) \pbart(z) \ud z,
\end{align*}
and since $\int \frac{\partial \pbart}{\partial t}\ud z = 0$ we have $\dot{\beta}_t = -\hat{h}$.

\textbf{Expressions for LQ case:}
Since the system is LTI, we have $\grad\cdot(\frac{1}{\theta}D - \Sigma) = 0$, and $\grad^2\cdot(\pbart\Sigma) = \tr(\Sigma\grad^2 \pbart)$. 
We postulate that $\pbart = \normal(n_t,S_t)$, with $n_t = 0$ for all $0 \le t \le T$ to verify that the expressions for $\vf_t,\correct_t$ in Table \ref{tb:soln-lq} satisfy \eqref{eq:Poisson}. Now for any trial vector field $\vftr_t(\cdot)$ we have
\begin{align*}
\nabla\cdot(\pbart \vftr_t) = \pbart\grad\cdot\vftr_t + \vftr_t\tp(\grad\pbart).
\end{align*}
To begin we observe $\grad \pbart = -\pbart S_t\inv (z - n_t)$, 
\begin{gather*}
\grad^2 \pbart = - \pbart S_t\inv + \pbart S_t\inv (z - n_t)(z - n_t)\tp S_t\inv, \\
\nabla \cdot(S_tC\tp C(z+n_t)) = S_tC\tp C \\
\nabla \cdot(\Sigma S_t\inv (z-n_t)) = \Sigma S_t\inv.
\end{gather*}
Moreover, since $\grad^2\log\pbart = S_t\inv$, and $\grad\cdot(Az) = \tr(A)$, we have that 
\begin{align*}
h_t(z) = 
\begin{cases}
\half |Cz|^2 \, + \, (\text{function of } t \text{ only}); & \text{LQG} \\ \noalign{\vskip5pt}
\frac{|\theta|}{2} |Cz|^2 \, + \, (\text{function of } t \text{ only}); & \text{LEQG} 
\end{cases}
\end{align*}
therefore, 
\begin{align*}
\hat{h}_t = 
\begin{cases}
\half \tr(C\tp C(S_t + n_t n_t\tp)); & \text{LQG} \\ \noalign{\vskip5pt}
\frac{|\theta|}{2}  \tr(C\tp C(S_t + n_t n_t\tp)); & \text{LEQG} 
\end{cases}
\end{align*}
Substituting into \eqref{eq:Poisson} we see that the equations are indeed satisfied. The Gaussian  hypothesis is also verified because the terminal condition of the mean-field process is Gaussian, and the equation for mean-field process becomes linear.   

\textbf{Proof of Proposition \ref{prop:Bsigma}:} substituting $\correct_t \equiv 0$ and the corresponding expression for $D_{\eta}$ in \eqref{eq:FPE}, and using the simplifications stated above for the LQ case, we see that the Fokker Planck equation \eqref{eq:FPE} matches the PDE for $p$ as given in Table \ref{tb:pt}.


\end{document}